\documentclass[10pt,a4paper,runningheads]{llncs}

\usepackage{times}
\usepackage[latin1]{inputenc}
\usepackage[english]{babel}
\usepackage{amsmath}
\usepackage{amssymb}

\usepackage{colortbl}
\usepackage{txfonts}
\usepackage{xypic}
\usepackage{multirow}
\usepackage{textcomp}
\usepackage{calc}
\usepackage{stmaryrd}
\usepackage{marvosym}
\usepackage{pgf}

\usepackage{ifthen}
\usepackage{mdwlist}
\usepackage{enumerate}
\usepackage{paralist}

\usepackage{tikz}
\usetikzlibrary{arrows,automata}
\usetikzlibrary{decorations.pathmorphing}
\usetikzlibrary{decorations.markings}
\usetikzlibrary{decorations.shapes}

\DeclareMathAlphabet{\mathpzc}{OT1}{pzc}{m}{it}

\tikzset{
	path/.style={dotted},
	every edge/.style={draw,solid},
	normal/.style={solid},
}

\newcommand{\mrk}[1]{{#1}'}
\newcommand{\oldstack}[3]{%
{\ifthenelse{\equal{#1}{1}}{%
\mrk{#2}
}%
{#2}}_{#3}%
}
\newcommand{\stack}[3]{%
[%
{\ifthenelse{\equal{#1}{1}}{%
\mrk{#2}
}%
{#2}}
{\ifthenelse{\equal{#1}{0}}{\ }{} }
{#3}%
]%
}

\newcommand{\va}[1]{\stackrel{#1}{\longrightarrow}}
\newcommand{\ourpath}[1]{\stackrel{#1}{\leadsto}}
\newcommand{\flush}[1]{\stackrel{#1}{\Longrightarrow}}

\newcommand{\chain}[3]{{}^{#1}\!\left[ #2 \right]\!{}^{#3}}

\newcommand{\config}[2]{\langle #1, \ #2 \rangle}
\newcommand{\symb}[1]{\mathop{smb}(#1)}
\newcommand{\state}[1]{\mathop{st}(#1)}

\newcommand{\comp}[1]{ \stackrel {{#1}} \vdash }

\newcommand{\avv}{\curvearrowright }

\DeclareMathOperator{\tree}{Tree}

\DeclareMathOperator{\fl}{Flush}
\DeclareMathOperator{\xz}{Succ}
\DeclareMathOperator{\avvi}{Next}

\newcommand{\flk}[3]{\fl_{#1}(#2,#3)}
\newcommand{\suc}[3]{\xz_{#1}(#2,#3)}
\newcommand{\nex}[3]{\avvi_{#1}(#2,#3)}

\newcommand{\ee}{e}

\begin{document}

\title{Logic Characterization of Floyd Languages}
\author{Violetta Lonati\inst{1}, Dino Mandrioli\inst{2}, Matteo Pradella\inst{2}}

\institute{
  DSI - Universit\`a degli Studi di Milano,
  via Comelico 39/41, Milano, Italy\\
  \email{lonati@dsi.unimi.it}
\and
  DEI - Politecnico di Milano,
  via Ponzio 34/5, Milano, Italy\\
  \email{\{dino.mandrioli, matteo.pradella\}@polimi.it}
}

\maketitle

\begin{abstract}
  Floyd languages (FL), alias Operator Precedence Languages, have
  recently received renewed attention thanks to their closure
  properties and local parsability which allow one to apply automatic
  verification techniques (e.g. model checking) and parallel and
  incremental parsing. They properly include various other classes,
  noticeably Visual Pushdown languages. In this paper we provide a
  characterization of FL in terms a monadic second order logic (MSO),
  in the same style as B\"uchi's one for regular languages. We prove the
  equivalence between automata recognizing FL and the MSO
  formalization.
\end{abstract}

\section{Introduction}
Floyd languages (FL), as we recently renamed Operator Precedence
Languages after their inventor, were originally introduced to support
deterministic parsing of programming and other artificial languages:
by taking inspiration from the structure of arithmetic expressions,
which gives precedence to multiplicative operations w.r.t. additive
ones, Robert Floyd defined an operator precedence matrix (OPM)
associated with a context-free (operator) grammar. When the OPM is
free of conflicts it is easy to build a deterministic shift-reduce
algorithm that associates any language sentence with a unique syntax
tree~\cite{Floyd1963}. FL and related grammars (FG) were also studied
with different motivations, such as grammar inference. This lead to
discover interesting closure properties that are not enjoyed by more
general context-free (CF) languages \cite{Crespi-ReghizziM12}. After
these initial results the interest in FL properties decayed for
several decades, probably due to the advent of more expressive
grammars, such as LR ones \cite{GruneJacobs:08} which also allow for
efficient deterministic parsing.

Recently, however, we revitalized our interest in FL on the basis of
two rather unexpected remarks.  First, and rather occasionally, we
noted that a newer class of CF deterministic languages, namely Visual
Pushdown Languages (VPL) -and other connected families
\cite{Berstel:2001:BGT,conf/mfcs/NowotkaS07,caucal:DSP:2008:1743}- are
a proper subclass of FL. VPL have been introduced and investigated
\cite{AluMad04} with the main motivation to extend to them the same or
similar automatic analysis techniques -noticeably, model checking-
that have been so successful for regular languages; their major
features which made them quite successful in the literature are that:
despite being recognized by infinite state machines -a specialized
class of pushdown automata- they enjoy practically all closure
properties exhibited by regular languages; they can be defined by a
suitable logic formalism that extends in a fairly natural way the
classical Monadic Second Order (MSO) logic characterization introduced
by B\"uchi for finite state automata \cite{thomas90}.  These features,
paired with the decidability of the emptiness problem shared by all CF
languages, makes them amenable for the application of typical model
checking techniques.  When we realized that VPL are subclass of FL
characterized by a well-precise ``shape'' of OPM we also investigated
other closure properties that were not yet known: by joining old
results of decades ago \cite{Crespi-ReghizziMM1978} with new ones
\cite{Crespi-ReghizziM12}, it turns out the FL enjoy the same closure
properties w.r.t. main operations such as Boolean ones, concatenation,
Kleene *, etc. as regular languages and VPL. Thus, FL too are amenable
for a significant extension of model checking techniques.
 
A second major motivation that renewed our interest in FL -which,
however, has a lesser impact on the present research- is their
locality principle, which makes them much better suited than other
deterministic CF languages for parallel and incremental (parsing)
techniques: unlike more general languages, in fact, the parsing of a
substring $w$ of a string $x$ can be carried over independently of the
``context'' of $w$ within $x$; we feel that in the era of multicore
machines the minor loss in expressive power of FG w.r.t. say, LR ones,
is far compensated by the gain of efficiency in -possibly incremental
analysis- that can be obtained by exploiting parallelism
\cite{BCMPPV12}.

In our path of ``rediscovering FL and their properties'', we also
filled up a fairly surprising hole in previous literature, namely the
lack of an automata family that perfectly matches FG in terms of
generative power: Floyd Automata (FA) are reported in \cite{LMP11}
and, with more details and precision, in \cite{LMP10}.

In this paper we provide the ``last tile of the puzzle'', i.e., a
complete characterization of FL in terms of a suitable MSO, so that,
as well as with regular languages and VPL, one can, for instance,
state a language property by means of a MSO formula; then
automatically verify whether a given FA accepts a language that enjoys
that property. Our new MSO logic is certainly inspired by the original
\cite{thomas90} approach, as well as the technique to automatically derive a
FA from a given formula; as it happened also with other previous
``extensions'' of properties and techniques to the FL family, however,
we had to face some new technical difficulties which sharply departed
from the original approaches of both regular and VPL \cite{thomas90},
\cite{jacm/AlurM09}. In this case the main difference between finite state
automata and VPA on one side and FA on the other one is that the
former ones are real-time machines -i.e. read an input character at
any move, whereas FA are not; thus, properties expressed in terms of
character positions cannot exploit the fact that to any position it
corresponds one and only one state of the automaton. In some sense the
logic formalization of a FL must encode the corresponding parsing
algorithm which is far from the trivial one of regular and VPL whose
strings have a shape isomorphic to the corresponding syntax tree.

The paper is structured as follows: Section~\ref{sec:prel} provides the
necessary background about FL and their
automata. Section~\ref{sec:logic} defines a MSO over strings and 
provides two symmetric constructions to derive an equivalent FA from a
MSO formula and conversely. Section~\ref{sec:concl} offers some
conclusion and hints for future work.

\section{Preliminaries}\label{sec:prel}

FL are normally defined through their generating grammars
\cite{Floyd1963,Fischer69}; in this paper, however, we
characterize them through their accepting automata
\cite{LMP10,LMP11} which are the natural way to state equivalence
properties with logic characterization. Nevertheless we assume some
familiarity with classical language theory concepts such as
context-free grammar, parsing, shift-reduce algorithm, syntax tree~\cite{GruneJacobs:08}.

Let $\Sigma = \{a_1, \dots, a_n \}$ be an alphabet. The empty string is denoted
$\epsilon$. 
We use a special symbol \# not in $\Sigma$ to mark the beginning and
the end of any string. This is consistent with the typical operator
parsing technique that requires the look-back and look-ahead of one
character to determine the next parsing action \cite{GruneJacobs:08}.

\begin{definition}\label{def:opm}
  An \textit{operator precedence matrix} (OPM) $M$ over an alphabet
  $\Sigma$ is a partial function $(\Sigma \cup \{\#\})^2 \to \{\lessdot,
  \doteq, \gtrdot\}$, that with each ordered pair $(a,b)$
  associates the OP relation $M_{a,b}$ holding between $a$ and $b$.  We
  call the pair $(\Sigma, M)$ an \emph{operator precedence alphabet}
  (OP).  Relations $\lessdot, \doteq, \gtrdot$, are named
  yields precedence, equal in precedence, takes precedence,
  respectively.
By convention, the initial \# can only yield precedence, and other
symbols can only take precedence on the ending \#.
\end{definition}

\noindent 
If $M_{a,b} = \circ$, where $\circ \in \{\lessdot, \doteq, \gtrdot \}$,
we write $a \circ b$.  For $u,v \in \Sigma^*$ we write $u \circ v$ if
$u = xa$ and $v = by$ with $a \circ b$.
$M$ is \emph{complete} if $M_{a,b}$ is defined for every $a$ and $b$ in $\Sigma$. 
Moreover in the following we assume that $M$ is {\em acyclic}, which
means that $c_1 \doteq c_2 \doteq \ldots \doteq c_k \doteq c_1$ does
not hold for any $c_1, c_2, \ldots c_k \in \Sigma, k\ge1$. See
\cite{Crespi-ReghizziMM1978,Crespi-ReghizziM12,LMP10}
for a discussion on this hypothesis.


\begin{definition}\label{def:fa}
A nondeterministic \emph{Floyd automaton} (FA) is a tuple
$\mathcal A = \langle \Sigma, M, Q, I, F, \delta \rangle $ where:
\begin{itemize}
\item $(\Sigma, M)$ is a precedence alphabet,
\item $Q$ is a set of states (disjoint from $\Sigma$),
\item $I, F \subseteq Q$ are sets of initial and final states, respectively, 
\item $\delta : Q \times ( \Sigma \cup Q) \rightarrow 2^Q$ is the transition function.
\end{itemize}
\end{definition}

\noindent The transition function is the union of two disjoint functions:
\[
\delta_{\text{push}}: Q \times \Sigma \rightarrow 2^Q
\qquad 
\delta_{\text{flush}}: Q \times Q \rightarrow 2^Q
\]
A nondeterministic FA can be represented by a graph
with $Q$ as the set of vertices and $\Sigma \cup Q$ as the set of edge
labelings: there is an edge from state $q$ to state $p$ labelled by
$a \in \Sigma$ if and only if $p \in \delta_{push}(q,a)$ and there is
an edge from state $q$ to state $p$ labelled by $r \in Q$ if and only
if $p \in \delta_{flush}(q,r)$.  To distinguish flush transitions from
push transitions we denote the former ones by a double arrow.

To define the semantics of the automaton, we introduce some notations.
We use letters $p, q, p_i, q_i, \dots $ for states in $Q$ and we set
$\mrk{\Sigma} = \{\mrk a \mid a \in \Sigma \}$; symbols in $\Sigma'$
are called {\em marked} symbols.  Let $\Gamma =  (\Sigma \cup \mrk
\Sigma \cup \{\#\}) \times Q$; we denote symbols in $\Gamma$ as
$\stack 0aq$, $\stack 1aq$, or $\stack 0\#q$, respectively.  We set
$\symb {\stack 0aq} = \symb {\stack 1aq} = a$, $\symb {\stack
  0\#q}=\#$, and $\state {\stack 0aq} = \state {\stack 1aq} = \state
{\stack 0\#q} = q$.

A \emph{configuration} of a FA is any pair 
$C = \config { B_1 B_2 \dots B_n } {a_1 a_2 \dots a_m }$,
where $B_i \in \Gamma$
and $a_i \in \Sigma\cup \{\#\}$.
The first component represents
the contents of the stack, while the second component is the part of input still
to be read.

A computation is a finite sequence of moves $C \vdash
C_1$; there are threeg kinds of moves, depending on the precedence
relation between $\symb {B_n}$ and $a_1$:

\smallskip

\noindent {\bf (push)} 
if $\symb {B_n} \doteq a_1$ then 
$
C_1 = \config{B_1 \dots B_n \stack 0{a_1}q }{a_2 \dots a_m}, \text{ with } q \in \delta_{push}(\state{B_n},a_1)$;
\smallskip

\noindent {\bf (mark)} 
if $\symb {B_n} \lessdot a_1$ then
$
C_1 = \config{B_1 \dots B_n \stack 1{a_1}q }{a_2 \dots a_m}, \text{ with } 
 q \in \delta_{push}(\state{B_n},a_1)$;

\smallskip

\noindent {\bf (flush)} 
if $\symb {B_n} \gtrdot  a_1$ then
let $i$ the greatest index such that
$\symb {B_i} \in \mrk\Sigma$. 
\[
C_1 = \config{B_1 \dots B_{i-2}\stack
0{\symb {B_{i-1}}}q}{ a_1 a_2 \dots a_m }, \text{ with } 
 q \in \delta_{flush}(\state{B_n}, \state{B_{i-1}}).
\]

\smallskip

Finally, we say that a configuration $\stack 0{\#}{q_I}$ is {\em
  starting} if $q_I \in I$ and a configuration $\stack 0{\#}{q_F}$ is
{\em accepting} if $q_F \in F$.  The language accepted by the
automaton is defined as:
\[
L(\mathcal A) = \left\{ x \mid  \config {\stack 0\#{q_I}} {x\#}  \comp * 
\config {\stack 0\#{q_F}} \# , q_I \in I, q_F \in F \right\}.
\]



Notice that transition function $\delta_{\text{push}}$ is used to
perform both push and mark moves.  To distinguish them, 
in the graphical
representation of a FA we will use a solid arrow to denote mark moves
in the state diagram.

The deterministic version of FA is defined along the usual lines.

\begin{definition}
  A FA is deterministic if $I$ is a singleton, and the ranges of
  $\delta_{\text{push}}$ and $\delta_{\text{flush}}$ are both $Q$
  rather than $2^Q$.
\end{definition}

In \cite{LMP10} we proved in a constructive way that nondeterministic FA have the same
expressive power as the deterministic ones and both are equivalent to the original Floyd grammars.

\begin{example}\label{ex:except}
  We define here the stack management of a simple programming language
  that is able to handle nested exceptions. For simplicity, there are
  only two procedures, called $a$ and $b$. Calls and returns are
  denoted by $call_a$, $call_b$, $ret_a$, $ret_b$, respectively.
  During execution, it is possible to install an exception handler
  $hnd$. The last signal that we use is $rst$, that is issued when an
  exception occur, or after a correct execution to uninstall the
  handler. With a $rst$ the stack is ``flushed'', restoring the state
  right before the last $hnd$.  
  Every $hnd$ not installed during the
  execution of a procedure is managed by the OS. We require also that
  procedures are called in an environment controlled by the OS, hence
  calls must always be performed between a $hnd$/$rst$ pair (in
  other words, we do not accept {\em top-level} calls). The automaton
  modeling the above behavior is
  presented in Figure \ref{ex:primo}.

Incidentally, notice that such a language is not a VPL but somewhat
extends their rationale: in fact, whereas VPL allow for unmatched
parentheses only at the beginning of a sentence (for returns) or at
the end (for calls), in this language we can have unmatched $call_a$, $call_b$,
$ret_a$, $ret_b$ within a pair $hnd$, $rst$.
\end{example}


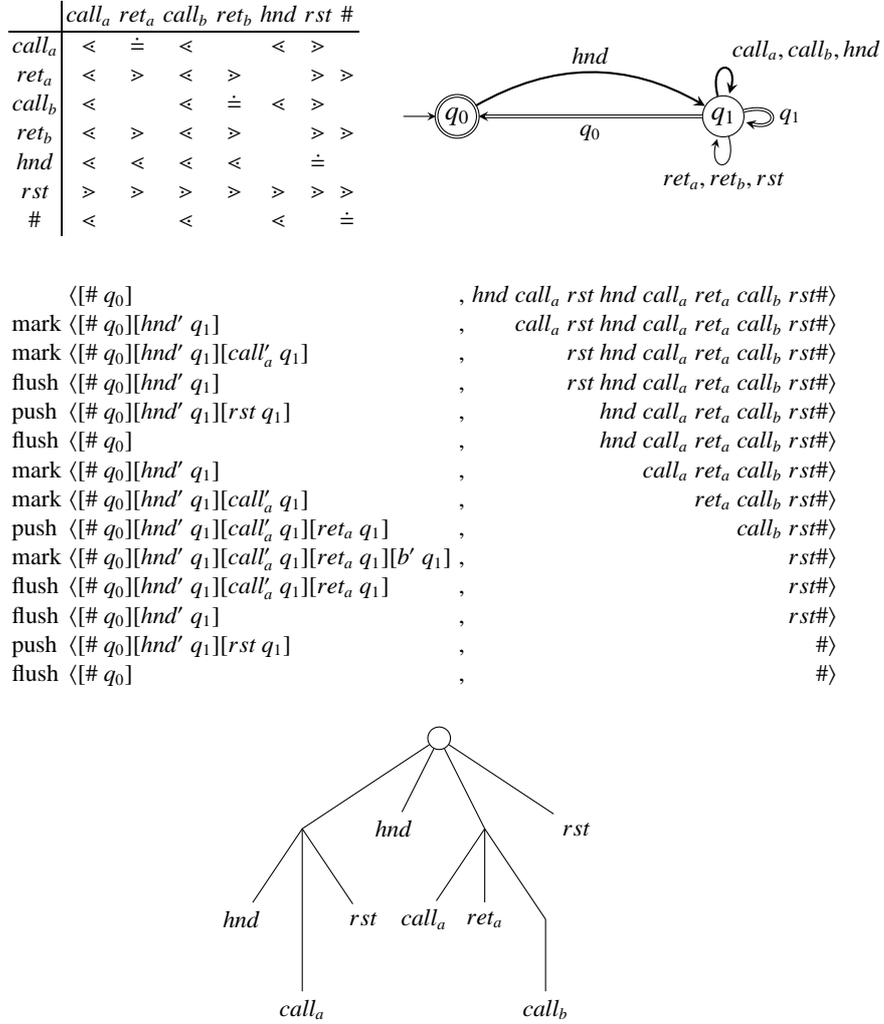
\begin{figure}
\begin{center}
\begin{tabular}{m{.4\textwidth}m{.45\textwidth}}
$
\begin{array}{c|ccccccc}
      & call_a & ret_a & call_b & ret_b & hnd & rst  &  \# \\
\hline
call_a & \lessdot & \dot= & \lessdot &  & \lessdot & \gtrdot  &  \\
ret_a  & \lessdot & \gtrdot & \lessdot & \gtrdot &  & \gtrdot  & \gtrdot \\
call_b & \lessdot & & \lessdot & \dot=  & \lessdot & \gtrdot  &  \\
ret_b  & \lessdot & \gtrdot & \lessdot & \gtrdot &  & \gtrdot  & \gtrdot \\
hnd    & \lessdot & \lessdot & \lessdot & \lessdot & & \dot=  &  \\
rst    & \gtrdot  & \gtrdot & \gtrdot & \gtrdot & \gtrdot  & \gtrdot & \gtrdot \\
\#     & \lessdot & & \lessdot & & \lessdot & & \dot= \\
\end{array}
$
&
\begin{tikzpicture}[every edge/.style={draw,solid}, node distance=4cm, auto, 
                    every state/.style={draw=black!100,scale=0.5}, >=stealth]

\node[initial by arrow, initial text=,state,accepting] (S) {{\huge $q_0$}};
\node[state] (E) [right of=S, xshift=3cm] {{\huge $q_1$}};

\path[->]
(S) edge [thick, bend left]  node {$hnd$} (E)
(E) edge [loop right, double] node {$q_1$} (E)
(E) edge [loop below] node {$ret_a, ret_b, rst$} (E)
(E) edge [thick, loop above, above right] node {$call_a, call_b, hnd$} (E)
(E) edge [below, double]  node {$q_0$} (S) ;
\end{tikzpicture}\\
\[
\begin{array}{llcr}
& \langle \stack 0 {\#}{q_0}      & , &    hnd\ call_a\ rst\ hnd\ call_a\ ret_a\ call_b\ rst \# \rangle \\
\text{mark} & \langle \stack 0 {\#}{q_0}  \stack 0 {hnd'}{q_1}    & , &   call_a\ rst\ hnd\ call_a\ ret_a\ call_b\ rst  \# \rangle \\
\text{mark} & \langle \stack 0 {\#}{q_0}  \stack 0 {hnd'}{q_1}\stack 0 {call_a'}{q_1}    & , &   rst\ hnd\ call_a\ ret_a\ call_b\ rst  \# \rangle \\
\text{flush} & \langle \stack 0 {\#}{q_0}  \stack 0 {hnd'}{q_1}    & , &    rst\ hnd\ call_a\ ret_a\ call_b\ rst   \# \rangle \\
\text{push} & \langle \stack 0 {\#}{q_0}  \stack 0 {hnd'}{q_1}\stack 0 {rst}{q_1}    & , &   hnd\ call_a\ ret_a\ call_b\ rst \# \rangle \\
\text{flush} & \langle \stack 0 {\#}{q_0}      & , &   hnd\ call_a\ ret_a\ call_b\ rst \# \rangle \\
\text{mark} & \langle \stack 0 {\#}{q_0}  \stack 0 {hnd'}{q_1}    & , &   call_a\ ret_a\ call_b\ rst \# \rangle \\
\text{mark} & \langle \stack 0 {\#}{q_0}  \stack 0 {hnd'}{q_1}\stack 0 {call_a'}{q_1}    & , &    ret_a\ call_b\ rst \# \rangle \\
\text{push} & \langle \stack 0 {\#}{q_0}  \stack 0 {hnd'}{q_1}\stack 0 {call_a'}{q_1}\stack 0 {ret_a}{q_1}    & , &   call_b\ rst \# \rangle \\
\text{mark} & \langle \stack 0 {\#}{q_0}  \stack 0 {hnd'}{q_1}\stack 0 {call_a'}{q_1}\stack 0 {ret_a}{q_1}\stack 0 {b'}{q_1}    & , &   rst  \# \rangle \\
\text{flush} & \langle \stack 0 {\#}{q_0}  \stack 0 {hnd'}{q_1}\stack 0 {call_a'}{q_1}\stack 0 {ret_a}{q_1}    & , &    rst \# \rangle \\
\text{flush} & \langle \stack 0 {\#}{q_0}  \stack 0 {hnd'}{q_1}    & , &    rst \# \rangle \\
\text{push} & \langle \stack 0 {\#}{q_0}  \stack 0 {hnd'}{q_1}\stack 0 {rst}{q_1}    & , &    \# \rangle \\
\text{flush} & \langle \stack 0 {\#}{q_0}      & , &    \# \rangle \\
\end{array}
\] \\
\end{tabular} 
\begin{tikzpicture}[scale=0.8]
\node [circle,draw]{}  child { [sibling distance=1cm] child { node {$hnd$} }  
                                           child { child {node {$call_a$}} }
                                           child { node {$rst$} }}
        child { node {$hnd$}  }
        child { [sibling distance=1cm]  child { node {$call_a$} }
                                          child { node {$ret_a$} }
                                          child { child {node {$call_b$}} }
                                        } 
        child { node {$rst$}};
\end{tikzpicture}
\caption{Precedence matrix, automaton, example run, and corresponding tree of Example~\ref{ex:except}.}\label{ex:primo}
\end{center}
\end{figure}

\begin{definition}\label{def:chain}
A \emph{simple chain} is a string $c_0 c_1 c_2 \dots c_\ell c_{\ell+1}$,
written as
$
\chain {c_0} {c_1 c_2 \dots c_\ell} {c_{\ell+1}},
$
such that:
$c_0, c_{\ell+1} \in \Sigma \cup \{\#\}$,
$c_i \in \Sigma$ for every $i = 1,2, \dots \ell$, 
and $c_0 \lessdot c_1 \doteq c_2 \dots c_{\ell-1} \doteq c_\ell \gtrdot c_{\ell+1}$.
A \emph{composed chain} is a string 
$c_0 s_0 c_1 s_1 c_2  \dots c_\ell s_\ell c_{\ell+1}$, 
where
$\chain {c_0}{c_1 c_2 \dots c_\ell}{c_{\ell+1}}$ is a simple chain, and
$s_i \in \Sigma^*$ is the empty string 
or is such that $\chain {c_i} {s_i} {c_{i+1}}$ is a chain (simple or composed),
for every $i = 0,1, \dots, \ell$. 
Such a composed chain will be written as
$\chain {c_0} {s_0 c_1 s_1 c_2 \dots c_\ell s_\ell} {c_{\ell+1}}$.

A string $s \in \Sigma^*$ is \emph{compatible} with the OPM $M$ if $\ \chain{\#}{s}{\#}$ is a chain.
\end{definition}

\begin{definition}
Let $\mathcal A$ be a Floyd automaton.
We call a \emph{support} for the simple chain
$\chain {c_0} {c_1 c_2 \dots c_\ell} {c_{\ell+1}}$
any path in $\mathcal A$ of the form
\begin{equation}
\label{eq:simplechain}
q_0
\va{c_1}{q_1}
\va{}{}
\dots
\va{}q_{\ell-1}
\va{c_{\ell}}{q_\ell}
\flush{q_0} {q_{\ell+1}}
\end{equation}

%
Notice that the label of the last (and only) flush is exactly $q_0$, i.e. the first state of the path; this flush is executed because of relation $c_\ell \gtrdot c_{\ell+1}$.

We call a \emph{support for the composed chain} 
$\chain {c_0} {s_0 c_1 s_1 c_2 \dots c_\ell s_\ell} {c_{\ell+1}}$
any path in $\mathcal A$ of the form
\begin{equation}
\label{eq:compchain}
q_0
\ourpath{s_0}{q'_0}
\va{c_1}{q_1}
\ourpath{s_1}{q'_1}
\va{c_2}{}
\dots
\va{c_\ell} {q_\ell}
\ourpath{s_\ell}{q'_\ell}
\flush{q'_0}{q_{\ell+1}}
\end{equation}
where, for every $i = 0, 1, \dots, \ell$: 

\begin{itemize}
\item if $s_i \neq \epsilon$, then $q_i \ourpath{s_i}{q'_i} $ 
is a support for the chain $\chain {c_i} {s_i} {c_{i+1}}$, i.e.,
it can be decomposed as $q_i\ourpath{s_i}{q''_i} \flush{q_i}{q'_i}$.

\item if $s_i = \epsilon$, then $q'_i = q_i$.
\end{itemize}
Notice that the label of the last flush is exactly $q'_0$.
\end{definition}

The chains fully determine the structure of the parsing of any
automaton over $(\Sigma, M)$. Indeed, if the automaton performs the computation
\[ 
\config{\stack 0 a {q_0}} {s b} \comp * 
\config{\stack 0 a {q}} {b} .
\]
then $\chain asb$ 
is necessarily a chain over $(\Sigma, M)$ and there exists a support
like \eqref{eq:compchain} with $s = s_0 c_1 \dots c_\ell s_\ell$ and $q_{\ell+1} = q$.

Furthermore, the above computation corresponds to the parsing by the
automaton of the string $s_0 c_1 \dots c_\ell s_{\ell}$ within the
context $a$,$b$. Notice that such context contains all
information needed to build the subtree whose frontier is that string.
This is a distinguishing feature of FL, not shared by other
deterministic languages: we call it the \emph{locality principle} of Floyd languages.

\begin{example}
With reference to the tree in Figure~\ref{ex:primo}, the
parsing of substring $hnd \ call_a \ rst \ hnd $ 
is given by computation
\[
\langle \stack 0 \# {q_0} \ ,\ hnd \ call_a \ rst \ hnd  \rangle
\comp{*}
\langle \stack 0 \# {q_0} \ ,\ hnd \rangle
\]
which corresponds to 
support 
$
q_0
\va{hnd}{q_1}
\va{call_a}{q_1}
\flush{q_1}{q_1}
\va{rst}{q_1}
\flush{q_0}{q_0}
$ of chain $\chain {\#} {hnd \ call_a \ rst} {hnd}$.
\end{example}

\begin{definition}\label{def:maxfa}
Given the OP alphabet $(\Sigma, M)$,  let us consider the FA
$\mathcal A(\Sigma, M)$ $=$ $\langle \Sigma, M,$ $\{q\}, \{q\}, \{q\}, \delta_{max}
\rangle $ where  $\delta_{max}(q,q) = q$, and $\delta_{max}(q,c) = q$,
$\forall c \in \Sigma$.
We call $\mathcal A(\Sigma, M)$ the \emph{Floyd Max-Automaton} over $\Sigma, M$.
\end{definition}

For a max-automaton $\mathcal A(\Sigma,M)$ each chain has a support;
since there is a chain $\chain{\#}{s}{\#}$ for any string $s$
compatible with $M$, a string is accepted by $\mathcal A(\Sigma, M)$
iff it is compatible with $M$.  Also, whenever $M$ is complete, each
string is compatible with $M$, hence accepted by the max-automaton. It
is not difficult to verify that a max-automaton is equivalent to a
max-grammar as defined in \cite{Crespi-ReghizziMM1978}; thus, when M
is complete both the max-automaton and the max-grammar define the
universal language $\Sigma^*$ by assigning to any string the (unique)
structure compatible with the OPM.

\medskip

In conclusion, given an OP alphabet, the OPM $M$ assigns a structure
to any string in $\Sigma^*$ compatible with $M$; a FA defined on the
OP alphabet selects an appropriate subset within such a
``universe''. In some sense this property is yet another variation of
the fundamental Chomsky-Sh\"utzenberger theorem.





\section{Logic characterization of FL}\label{sec:logic}

We are now ready to provide a characterization of FL in terms of a
suitable Monadic Second Order (MSO) logic in the same vein as
originally proposed bu B\"uchi for regular languages and subsequently
extended by Alur and Madhusudan for VPL. The essence of the approach
consists in defining language properties in terms of relations between
the positions of characters in the strings: first order variables are
used to denote positions whereas second order ones denote subsets of
positions; then, suitable constructions build an automaton from a
given formula and conversely, in such a way that formula and
corresponding automaton define the same language.  The extension
designed by~\cite{jacm/AlurM09} introduced a new basic binary predicate
$\leadsto$ in the syntax of the MSO logic, $x \leadsto y$ representing
the fact that in positions $x$ and $y$ two matching parentheses
--named call and return, respectively in their terminology-- are
located.  In the case of FL, however, we have to face new problems.
\begin{itemize}
\item Both finite state automata and VPA are real-time machines, i.e.,
  they read one input character at every move; this is not the case
  with more general machines such as FA, which do not advance the
  input head when performing flush transitions, and may also apply
  many flush transitions before the next push or mark which are the
  transitions that consume input. As a consequence, whereas in the
  logic characterization of regular and VP languages any first order
  variable can belong to only one second order variable representing
  an automaton state, in this case --when the automaton performs a
  flush-- the same position may correspond to different states and
  therefore belong to different second-order variables.
\item In VPL the $\leadsto$ relation is one-to-one, since any call
  matches with only one return, if any, and conversely. In FL, instead
  the same position $y$ can be ``paired'' with different positions $x$
  in correspondence of many flush transitions with no push/mark in
  between, as it happens for instance when parsing a derivation such
  as $A \stackrel{*}{\Rightarrow} \alpha^k A$, consisting of $k$ immediate
  derivations $A \Rightarrow \alpha A$; symmetrically the same
  position $x$ can be paired with many positions $y$.
\end{itemize}
In essence our goal is to formalize in terms of MSO formulas a
complete parsing algorithm for FL, a much more complex algorithm than
it is needed for regular and VP languages.  The first step to achieve
our goal is to define a new relation between (first order variables
denoting) the positions in a string.

In some sense the new relation formalizes structural properties of FL
strings in the same way as the VPL $\leadsto$  relation does for VPL; the new
relation, however, is more complex as its VPL counterpart in a
parallel way as FL are much richer than VPL.

\begin{definition} 
\label{def:avv}
Consider a string $s \in \Sigma^*$ and a OPM $M$.  For $0 \le x < y
\le |s|+1$, we write $x \avv y $ iff there exists a sub-string of $\#
s\#$ which is a chain $\chain{a}{r}{b}$,
such that $a$ is in position $x$ and $b$ is in position
$y$.
\end{definition}

\begin{example}
  With reference to the string of Figure~\ref{ex:primo}, we have $1
  \avv 3$, $0 \avv 4$, $6 \avv 8$, $4 \avv 8$, and $0 \avv 9$.  Notice
  that, in the parsing of the string, such pairs correspond to
  contexts where a reduce operation is executed (they are listed
  according to their execution order).
\end{example}
In general $x \avv y$ implies $y > x+1$, and a position $x$ may be in
such a relation with more than one position and vice versa.  Moreover,
if $s$ is compatible with $M$, then $0 \avv |s|+1$.


\subsection{A Monadic Second-Order Logic over Operator Precedence Alphabets}

Let ($\Sigma$,$M$) be an OP alphabet. According to Definition~\ref{def:avv} it induces the relation
$\avv$ over positions of characters in any words in $\Sigma^*$.  Let us define a countable infinite set
of first-order variables $x, y, \dots$ and a countable infinite set of
monadic second-order (set) variables $X,Y, \dots$.


\begin{definition}
	The MSO$_{\Sigma,M}$ (\emph{monadic second-order logic} over $(\Sigma,
	M)$) is defined by the following syntax:
\[
	\varphi := 
		a(x) \mid 
		x \in X \mid
		x \leq y \mid
		x \avv y \mid
                x = y+1 \mid
		\neg \varphi \mid
		\varphi \lor \varphi \mid
		\exists x.\varphi \mid
		\exists X.\varphi 
\]
where $a \in \Sigma$, $x, y$ are first-order variables and $X$ is a set variable.
\end{definition}

MSO$_{\Sigma,M}$ formulae are interpreted over $(\Sigma, M)$ strings and the positions of their characters in the following natural way: 
\begin{itemize}
\item
first-order variables are interpreted over positions of the string; 
\item
second-order variables are interpreted over sets of positions;
\item
$a(x)$ is true iff the character in position $x$ is $a$;
\item
$x \avv y $ is true iff $x$ and $y$ satisfy Definition~\ref{def:avv};
\item
the other logical symbols have the usual meaning.
\end{itemize}



A sentence is a formula without free variables.
The language of all strings $s \in \Sigma^*$ such that $\# s \# \models \varphi$ is denoted by
$L(\varphi)$:
\[
L(\varphi) = \{ s\in \Sigma^* \mid \# s \# \models \varphi \}
\]
where $\models$ is the standard satisfaction relation.

\begin{example}
\label{ex:ss}
Consider the language of Example~\ref{ex:except}, with the structure
implied by its OPM. 
The following sentence defines it:
\[
\forall z \left(
\left( 
\begin{array}{c}
call_a(z) \lor ret_a(z) \\
\lor \\
call_b(z) \lor ret_b(z)         
\end{array}
\right)
\Rightarrow
\exists x, y
\left( 
	\begin{array}{c}
	x \avv y  \land 
        x < z < y  \\
\land \\
        hnd(x+1) \land rst(y-1)
	\end{array}
\right)
\right).
\]
\end{example}

\begin{example}
\label{ex:tt}
Consider again Example~\ref{ex:except}. If we want to add the
additional constraint that procedure $b$ cannot directly install handlers
(e.g. for security reasons), we may state it through the following
formula:
\[
\forall z 
\left(
hnd(z)
\Rightarrow 
\neg \exists u \left(
call_b (u) \ \land \
(u+1 = z \lor u \avv z )
\right)
\right)
\]

\end{example}

We are now ready for the main result.

\begin{theorem}
\label{teo:logic}
A language $L$ over $(\Sigma,M)$ is a FL if and only if there exists a
MSO$_{\Sigma, M}$ sentence $\varphi$ such that $L= L(\varphi)$.
\end{theorem}
The proof is constructive and structured in the following two subsections.

\subsection{From MSO$_{\Sigma, M}$ to Floyd automata}

\begin{proposition}
\label{teo:logic:f2a}
Let $(\Sigma, M)$ be an operator precedence alphabet and  $\varphi$ be a MSO$_{\Sigma, M}$
sentence. Then $L(\varphi)$ can be recognized by a Floyd automaton over
$(\Sigma, M)$.
\end{proposition}

\begin{proof}
The proof follows the one by Thomas~\cite{thomas90} and is composed of two steps:
first the formula is rewritten so that
no predicate symbols nor first order variables are used;
then an equivalent FA is built inductively.
	
Let $\Sigma$ be $\{a_1, a_2, \dots,a_n \}$. For each predicate symbol $a_i$ we
introduce a fresh set variable $X_i$, therefore formula $a_i(x)$ will be
translated into $x \in X_i$.
Following the standard construction of~\cite{thomas90}, we also translate every first order variable into a fresh
second order variable with the additional constraint that the set it represents
contain exactly one position.


Let $\varphi'$ be the formula obtained from $\varphi$ by such a translation,
and consider any subformula $\psi$ of $\varphi'$:
let $X_1, X_2, \dots, X_n, X_{n+1}, \dots X_{n+m(\psi)}$ be the (second order)
variables appearing in $\psi$. 
Recall that $X_1, \dots, X_n$ represent
symbols in $\Sigma$, hence they are never quantified. 

As usual we interpret formulae over strings; in this case we use the alphabet 
\[
\Lambda(\psi)  = \left\{ \alpha \in  \{0,1\}^{n+m(\psi)} \mid 
	\exists !i  \text{ s.t. } 1 \le i \le n, \ \alpha_i = 1
	\right\} 
\]
A string $w \in \Lambda(\psi)^*$, with $|w| = \ell$, is used to interpret
$\psi$ in
the following way: the projection over $j$-th component of $\Lambda(\psi)$ gives  an
evaluation $\{1,2,\dots, \ell\} \to \{0,1\}$ of $X_j$, 
for every $1 \leq j \le n+m(\psi)$.

For any $\alpha \in \Lambda(\psi)$, the projection of $\alpha$ over the first $n$ components encodes a symbol in
$\Sigma$, denoted as $symb(\alpha)$. 
The matrix $M$ over $\Sigma$ can be naturally extended to the OPM
$M(\psi)$ over $\Lambda(\psi)$ by defining ${M(\psi)}_{\alpha, \beta} = M_{symb(\alpha), symb(\beta)}$
for any $\alpha, \beta \in \Lambda(\psi)$.

We now build a FA $\mathcal A$ equivalent to $\varphi'$. The construction is inductive on the structure of the formula: first we define the FA for all atomic formulae. We give here only the construction for $\avv$, since for the other ones the construction is standard and is the same as in~\cite{thomas90}.

Figure~\ref{fig:buchi} represents the Floyd automaton for atomic formula $\psi = 
X_i \avv X_j$ (notice that $i, j > n$).
For the sake of brevity, we use notation $[X_i]$ to represent the set of all tuples
$\Lambda(\psi)$ having the $i$-th component equal to 1; notation $[\bar{X}]$
represents the set of all tuples in $\Lambda(\psi)$ having both $i$-th and $j$-th
components equal to 0.
The automaton, after a generic sequence of moves corresponding to
visiting an irrelevant portion of the syntax tree, when reading $X_i$
performs either a mark or a push move, depending on whether $X_i$ is a
leftmost leaf of the tree or not; then it visits the subsequent
subtree ending with a flush labeled $q_1$; at this point, if it reads
$X_j$, it accepts anything else will follow the examined fragment.

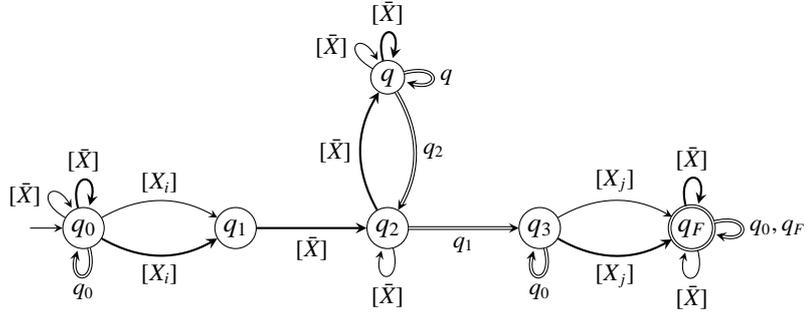
\begin{figure}
\begin{center}
\begin{tikzpicture}[every edge/.style={draw,solid}, node distance=4cm, auto, 
                    every state/.style={draw=black!100,scale=0.5}, >=stealth]

\node[initial by arrow, initial text=,state] (Q0) {{\huge $q_0$}};
\node[state] (Q1) [right of=Q0, xshift=0cm] {{\huge $q_1$}};
\node[state] (Q2) [right of=Q1, xshift=0cm] {{\huge $q_2$}};
\node[state] (Q)  [above of=Q2, xshift=0cm] {{\huge $q$}};
\node[state] (Q3) [right of=Q2, xshift=0cm] {{\huge $q_3$}};
\node[state] (QF) [right of=Q3, xshift=0cm, accepting] {{\huge $q_F$}};

\path[->]

(Q1) edge [thick, below]  node {$[\bar{X}]$} (Q2)
(Q2) edge [thick, left, bend left]  node {$[\bar{X}]$} (Q)
(Q)  edge [double, right, bend left]  node {$q_2$} (Q2)

(Q0) edge [bend left, above]  node {$[X_i]$} (Q1) 
(Q0) edge [bend right, below, thick]  node {$[X_i]$} (Q1) 

(Q2) edge [double, below]  node {$q_1$} (Q3) 

(Q0) edge [in=120, out=150, loop, left] node {$[\bar{X}] $} (Q0)
(Q0) edge [thick, loop above, above] node {$[\bar{X}]$} (Q0)
(Q0) edge [double, loop below] node {$q_0$} (Q0)

(Q2) edge [loop below] node {$[\bar{X}]$} (Q2)

(Q) edge [in=120, out=150, loop, left] node {$[\bar{X}]$} (Q)
(Q) edge [thick, loop above, above] node {$[\bar{X}] $} (Q)
(Q) edge [double, right, loop right]  node {$q$} (Q) 

(Q3) edge [bend left, above]  node {$[X_j]$} (QF) 
(Q3) edge [bend right, below, thick]  node {$[X_j]$} (QF) 
(Q3) edge [double, loop below] node {$q_0$} (Q3)

(QF) edge [loop below] node {$[\bar{X}] $} (QF)
(QF) edge [thick, loop above, above] node {$[\bar{X}]$} (QF)
(QF) edge [double, right, loop right]  node {$q_0, q_F$} (QF) 
;

\end{tikzpicture}
\caption{Floyd automaton for atomic formula $\psi = X_i \avv X_j$ }\label{fig:buchi}
\end{center}
\end{figure}

Then, a natural inductive path leads to the construction of the automaton associated with a generic MSO formula:
the disjunction of two subformulae can be obtained by building the
union automaton of the two corresponding automata; similarly for negation.
The existential quantification of $X_i$ is obtained by projection erasing the
$i$-th component. Notice that all matrices $M(\psi)$ are well defined for any
$\psi$ because the first $n$ components of the alphabet are never erased by
quantification.
The alphabet of the automaton equivalent to $\varphi'$ is $\Lambda(\varphi')=\{0,1\}^n$, which is
in bijection with $\Sigma$. 
\end{proof}

\subsection{From Floyd automata to MSO$_{\Sigma, M}$}

Let $\mathcal A$ be a deterministic Floyd automaton over
$(\Sigma,M)$. We build a MSO$_{\Sigma, M}$ sentence $\varphi$ such
that $L (\mathcal A)= L(\varphi)$.
The main idea for encoding the behavior of the Floyd automaton is
based on assigning the states visited during its run to positions
along the same lines stated by B\"uchi \cite{thomas90} and
extended for VPL \cite{jacm/AlurM09}. Unlike finite state automata and VPA, however, Floyd automata
do not work on-line.  Hence, it is not possible to assign a single
state to every position.  Let $Q = \{q_0, q_1, \ldots, q_N\}$ be the
states of $\mathcal A$ with $q_0$ initial; as usual, we will use second order variables
to encode them.  We shall need three different sets of second order
variables, namely $P_0, P_1, \ldots, P_N$, $M_0, M_1, \ldots, M_N$ and
$F_0, F_1, \ldots, F_N$: set $P_i$ contains
those positions of $s$ where state $i$ may be assumed \emph{after} a
push transition.
$M_i$ and $F_i$ represent the state reached after a flush: $F_i$
contains the positions where the flush occurs, whereas $M_i$ contains
the positions preceding the corresponding mark. Notice that any position
belongs to one only $P_i$, whereas it may belong to several $F_i$ or
$M_i$ (see Figure~\ref{fig:MP}).

\begin{figure}
\begin{center}
\begin{tabular}{ccc}
\begin{tikzpicture}[scale=0.7]
\node [circle,draw]{}   child { node {$t \in M_1 \cap M_2$} }  
                        child { child { child 
                                        child { node {$w \in F_1$} }}
                                child { node {$z \in F_2$}}}; 
\end{tikzpicture} & \qquad \qquad \qquad &
\begin{tikzpicture}[scale=0.7]
\node [circle,draw]{}   child { node {$w \in M_1$} }  
                        child { child { node {$z \in M_2$}}
                                child { child 
                                        child { node {$t \in F_1 \cap F_2$} }}
                                }; 
\end{tikzpicture}
\end{tabular}
\end{center}
\caption{Example trees with a position $t$ belonging to more than
  one $M_i$ (left) and $F_i$ (right).}\label{fig:MP}
\end{figure}
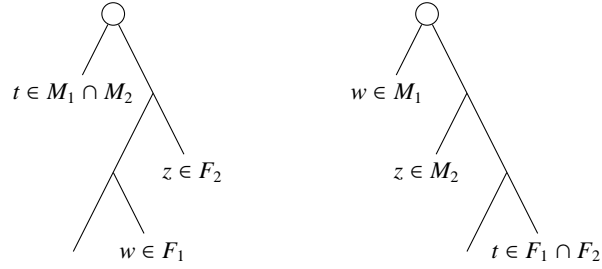

We show that $\mathcal A$ accepts a string $s$ iff $\# s \# \ \models
\ \varphi$, where
\begin{equation}
\label{eq:logic}
\begin{array}{lll}
\varphi & := & 
\exists 
P_0, P_1, \ldots, P_N,
M_0, M_1, \ldots, M_N,
F_0, F_1, \ldots, F_N, \ee
\quad \varphi'  \\
\varphi'  & := & 
0 \in P_0 
 \land 
\bigvee_{i \in F} \ee \in F_i
 \land 
\neg \exists x (\ee + 1< x) 
 \land 
\#(\ee+1)
 \land 
\varphi_\delta
\land 
\varphi_{exist}
 \land 
\varphi_{unique}.
\end{array}
\end{equation}
The first clause in $\varphi'$ encodes the initial state, whereas the
second, third and fourth ones encode the final states. We use
variable $\ee$ to refer to the \emph{end} of $s$, i.e., $\ee$ equals
the last position $|s|$.
%
The remaining clauses are defined in the following:
the fourth one encodes the transition function;
the last ones together encode the fact that there exists exactly one state that may
be assumed by a push transition in any position,
and the correspondence between mark and flush transitions. 

For convenience we introduce in formulae precedence relations and other
shortcut notations, presented next.

\noindent\textit{Notation.} In the following, when considering a chain $\chain a s b$ we assume $s = s_0 c_1 s_1 \dots c_\ell s_\ell$,
with $\chain a {c_1 c_2 \dots c_\ell} b$ a simple chain (any $s_g$ may be empty).
Also let $x_g$ be the position of symbol $c_g$, for $g = 1, 2, \ldots, \ell$
and, for the sake of uniformity, set $c_0 = a$, $x_0 = 0$,  $c_{\ell+1} = b$, and $x_{\ell+1} = |s|+1$.

\begin{eqnarray*}
x \circ y &:=&  \bigvee_{M_{a,b} = \circ} a(x) \land b(y), 
\text{ for }  
\circ \in \{\lessdot, \doteq, \gtrdot \}\\
\tree (x,z,w,y) &:=& 
\left(
\begin{array}{c}
x \avv y 
\\ \land \\
 ( x+1 = z \ \lor\  x \avv z ) \land \neg \exists t ( x < t < z \land x \avv t ) 
\\ \land \\
 ( w+1 = y \ \lor\ w \avv y ) \land \neg \exists t ( w < t < y
\land w \avv y ) 
\end{array}
\right) \\
\suc k x y & := & 
x+1 = y \land x \in P_k
\\
\nex k x y & := &
x \avv y \land x \in M_k \land y-1 \in F_k
\end{eqnarray*}
\begin{eqnarray*}
\flk k x y &:=& 
x \avv y \ \land \ x \in M_k \ \land y-1 \in F_k \ \land \\
& &  
\exists z,w \left(
\tree(x,z,w,y) \land 
\bigvee_{i=0}^N \bigvee_{j=0}^N 
\left(
\begin{array}{c}
\delta(q_i, q_j) = q_k  \\
\land \\ 
(\suc i w y \lor \nex i w y )\\
\land \\
 (\suc j x z \lor \nex j x z )
\end{array}
\right)
\right)\\
\tree_{i,j}( x,z,w,y) &:=&
\tree(x,z,w,y) \land
\left(
\begin{array}{c}
\suc i w y \lor \flk i w y )\\
\land \\ 
 (\suc j x z \lor \flk j x z )\\
\end{array}
\right)
\end{eqnarray*}

\noindent {\it Remarks. } If $x \avv y$ then there exist (unique) $z$ and $w$ such that $\tree(x,z,w,y)$ holds. 
In particular,  if $\chain a s b$ is a simple chain, then $0 \avv \ell+1$ and $\tree(0, 1, \ell, \ell+1)$ holds;
if $\chain a s b$ is a composed chain, 
then $0 \avv |s|+1$ and   $\tree(0, x_1, x_\ell, x_{\ell+1})$ holds. If $s_0 = \epsilon$ then $x_1 = 1$,
and if $s_\ell = \epsilon$ then $x_\ell = |s|$.

\noindent By definition, $\tree_{i,j} (x,z,w,y) \land q_k = \delta(q_i,q_j)$ implies $\flk k x y$.

\noindent  If $\chain a {c_1 c_2 \dots c_\ell} b$
is a simple chain with support
\begin{equation}
\label{eq:simplechainq}
q_i = 
q_{t_0}
\va{c_1}{q_{t_1}}
\va{c_2}{}
\dots
\va{c_\ell}{q_{t_\ell}}
\flush{q_{t_0}}{q_k}
\end{equation}
then $\tree_{t_0,t_\ell}( 0,1,\ell, \ell+1) $
and $\flk k 0 {\ell+1}$ hold; 
if 
$\chain a {s_0 c_1 s_1 c_2 \dots c_\ell s_\ell} b$
is a composed chain with support
\begin{equation}
\label{eq:compchainq}
q_i = 
q_{t_0}
\ourpath{s_0}{q_{f_0}}
\va{c_1}{q_{t_1}}
\ourpath{s_1}{q_{f_1}}
\va{c_2}{}
\dots
\va{c_g}{}
{q_{t_g}}\ourpath{s_g}{q_{f_g}}
\dots
\va{c_\ell}{q_{t_\ell}}
\ourpath{s_\ell}{q_{f_\ell}}
\flush{q_{f_0}}{q_k}
\end{equation}
then by induction we can see that $\tree_{f_\ell,f_0}( 0,|s_0|+1,|s_0\dots c_\ell| ,|s|+1) $ and
$\flk k 0 {|s|+1}$ hold.


\medskip

Formula $\varphi_\delta$ is the conjunction of the following formulae, organized in \emph{forward} formulae
and \emph{backward} formulae:

\medskip
\noindent {\em Forward formulae.}
\[
\varphi_{push\_fw} := \forall x, y  \bigwedge_{i=0}^{N} 
\left(
\begin{array}{c}
( x \lessdot y  \lor  x \doteq y)  \land a(y) \\
\land \\
\suc i x y \lor
\flk i x y
\end{array}
\Rightarrow y \in P_{\delta(q_i, a)}
\right)
\]
\[
\varphi_{flush\_fw} := 
\forall x, z, w, y
\bigwedge_{i=0}^N \bigwedge_{j=0}^N  
\left(
\tree_{i,j}(x,z,w,y)
\Rightarrow
	\begin{array}{c}
	x \in M_{\delta(q_i, q_j)} \\
	\land \\
	 y-1 \in F_{\delta(q_i, q_j)}
	\end{array}\right)
\]
\medskip
\noindent {\em Backward formulae.}
\[
\varphi_{push\_bw1} := \forall x,y  \bigwedge_{k=0}^{N} 
\left(
\begin{array}{c}
( x \lessdot y  \lor  x \doteq y)  \land a(y) \\
\land \\
y \in P_k \ \land \ x+1=y  
\end{array}
\Rightarrow
\bigvee_{i=0}^N  \left( \suc i x y 
\land \delta(q_i,a) = q_k \right)
 \right)
\]
\[
\varphi_{push\_bw2} := \forall x,y  \bigwedge_{k=0}^{N} 
\left(
\begin{array}{c}
( x \lessdot y  \lor  x \doteq y)  \land a(y) \\
\land \\
y \in P_k \ \land \ x \avv y  
\end{array}
\Rightarrow
\bigvee_{i=0}^N  \left( \flk i x y  \land \delta(q_i,a) = q_k \right)
 \right)
\]
\[
\varphi_{flush\_bwM} := 
\forall x \bigwedge_{k=0}^N \left( x \in M_k \Rightarrow 
\exists y,z,w 
\bigvee_{i=0}^N \bigvee_{j=0}^N  
\left(
\begin{array}{c} 
	\tree_{i,j}(x,z,w,y) \\
 	\land \\
 	\delta(q_i,q_j) = q_k
\end{array}
\right)
\right)
\]
\[
\varphi_{flush\_bwF} := 
\forall y \bigwedge_{k=0}^N \left( y \in F_k \Rightarrow 
\exists x,z,w 
\bigvee_{i=0}^N \bigvee_{j=0}^N  
\left(
\begin{array}{c} 
	\tree_{i,j}(x,z,w,y) \\
 	\land \\
 	\delta(q_i,q_j) = q_k
\end{array}
\right)
\right)
\]
\[
\varphi_{flush\_bw} := 
\forall x, z, w, y
\bigwedge_{k=0}^N  \bigwedge_{i=0}^N \bigwedge_{j=0}^N  
\left(\begin{array}{c} 
	\tree_{i,j}(x,z,w,y) \\
	 \land \\ 
	 \flk k x y \\
\end{array}
\Rightarrow
	\delta (q_i, q_j) = q_k
\right)
\]
\noindent
Formula $\varphi_{exist}$ is the conjunction of the following formulae:
\[
\varphi_{push\_exist} :=  
\forall x \left( \bigvee_{i=0}^{N} x \in P_i \right)
\]
\[
\varphi_{flush\_exist} :=  \forall x, y \left( x \avv y \Rightarrow 
\left(
\bigvee_{k=0}^{N} \flk k x y 
\right) \right)
\]
\noindent
Formula $\varphi_{unique}$ is the conjunction of the following formulae:
\[
\varphi_{push\_unique} :=  \forall x 
\bigwedge_{i=0}^{N}
\left(
x \in P_i \Rightarrow 
\neg 
\bigvee_{j=0}^{N} ( j \neq i \land x \in P_j )
\right)
\]
\[
\varphi_{flush\_unique} :=  \forall x, y 
\bigwedge_{k=0}^{N}
\left(
\flk k x y \Rightarrow 
\neg 
\bigvee_{j=0}^{N} ( j \neq k \land \flk j x y )
\right)
\]



\begin{remark}
\label{rem:unique}
If \eqref{eq:logic} holds, then for each $x,y$ $\suc i x y \lor \flk i x y$ implies that such $i$ is unique.
Indeed, $\suc j x y$ and $\flk k x y$ are mutually exclusive;
if $\flk i x y$ then such $i$ is unique by $\varphi_{flush\_unique}$;
if $\suc i x y$ then $y = x+1$ and $x \in P_i$, thus such $i$ is unique by $\varphi_{push\_unique}$.
\end{remark}


Now let $\mathcal C = \chain asb$ be a chain in $(\Sigma, M)$
and set
\[
\psi_{i,k} \ := \ 
\begin{array}{c}
\exists P_0, P_1, \ldots, P_N \\
\exists M_0, M_1, \ldots, M_N \\
\exists F_0, F_1, \ldots, F_N \
\end{array}
\
\exists \ee
\
\left(
0 \in P_i 
\ \land \
\flk k 0 {\ee+1}
\ \land \
\varphi_{\delta}
\ \land \
\varphi_{exist}
\ \land \
\varphi_{unique}
\right).
\]

The following lemmata hold.

\begin{lemma}\label{lem:a2f}
  If there exists a support $q_i \ourpath{s}{q_k} $ for the chain
  $\mathcal C$ in $\mathcal A$, then $asb \models \psi_{i, k}$.
\end{lemma}

\begin{proof}
We prove the lemma by induction on the structure of chains. 


\noindent {\bf Base step }
Let $\mathcal C$ be a simple chain and its support be decomposed  as
in~\eqref{eq:simplechainq}.

Define $\ee = \ell$, and $P_0, P_1,\ldots, P_N, M_0, \ldots, M_N, F_0, \ldots, F_N $ as follows.
$M_h$ is empty except for $M_k = \{0\}$;
$F_h$ is empty except for $F_k = \{\ell\}$;
for every $x = 0 \ldots \ell$, let $P_h$ contain $x$ iff $t_x = h$
(i.e., $x \in P_{t_x}$);
finally let $P_{\delta(q_k,b)}$ contain $\ell +1$ if $a \lessdot b$ or $a \doteq b$.

Then  we show that  $\psi_{i,k}$ is satisfied
by checking every subformula in $\varphi_\delta$,
$\varphi_{exist}$, $\varphi_{unique}$.
\begin{enumerate}
\item $\varphi_{push\_fw}$ is satisfied $\forall x = y-1 < \ell$ with
  $y \in P_{\delta(q_x, a)} \land a(y)$.
Then $\delta(q_{t_\ell}, q_{t_0}) = q_k$ guarantees $\flk k 0
{\ell+1}$; and $\delta(q_{k}, b) = q_{t_{\ell+1}}$ guarantees $\ell + 1 \in
P_{\ell+1}$.

{\em Remark. } Even if $\mathcal A$ is deterministic, some
chains could have different supports. However, every support produces
exactly one assignment $P_{t_0}, P_{t_1}, \ldots, F_k, M_k$ that
satisfies $\psi_{t_0, k}$.

\item $\varphi_{flush\_fw}$ is satisfied for $x=0, z=1, w=\ell,
  y=\ell+1$
with 
  $P_{t_{\ell}}, P_{t_0}, F_k, M_k$ (for all other cases, it is $\neg
  \tree_{i,j} (x,z,w,y)$).

\item $\varphi_{push\_bw1}$ is satisfied in the natural way for every
  $y \le \ell$; for $y = \ell+1$, it is $x \gtrdot y$, $x+1=y$, which
  implies $\neg( x \lessdot y \lor x \doteq y)$ and the antecedent is
  false.

\item $\varphi_{push\_bw2}$, for every pair $(x,y) \ne (0,\ell+1)$ is
  satisfied with $\neg x \avv y$; for $x = 0$, $y =\ell+1$, if $x \gtrdot
  y$ the antecedent is false, otherwise it is satisfied with $\flk k 0
  {\ell+1}$, $P_{t_{\ell+1}}$.

\item $\varphi_{flush\_bwM}$ and $\varphi_{flush\_bwF}$ are satisfied
  with  $x=0$ and $y = \ell+1$, respectively.
(For $x>0$, $y\le\ell$ the antecedents are false.)

\item $\varphi_{flush\_bw}$ is satisfied in a vacuous way (false
  antecedent) for $(x,y) \ne (0, \ell+1)$. For $x = 0, y = \ell+1$ it
  is satisfied with $i = t_\ell, j=t_0, F_k$. 

\item $\varphi_{push\_exist}$, $\varphi_{push\_unique}$,
  $\varphi_{flush\_exist}$, and $\varphi_{flush\_unique}$ are
  always satisfied, because a) the chain has a support, b) $\mathcal
  A$ is deterministic.

\item $\psi_{t_0,k}$ is finally satisfied with $\flk k 0 {\ell+1}$.
\end{enumerate}

\smallskip

\noindent {\bf Induction step }

Let now $\mathcal C$ be a composed chain and let its support be decomposed as
in~\eqref{eq:compchainq}. 
Let us consider the case $s_0 \ne \epsilon \ne
s_\ell$ (other cases are similar and simpler, therefore
omitted). Thus, $\delta(q_{f_\ell}, q_{f_0}) = q_k$.

Let $\ee$ be $|s|$. By the inductive hypothesis,  for every $g = 0, 1, \ldots, \ell$ such that $s_g
\neq \epsilon$ we have
$c_g s_g c_{g+1} \models \psi_{t_g,f_g}$: let ${P_0}^g,\ldots, {P_N}^g, {M_0}^g, \ldots, {M_N}^g,
{F_0}^g, \ldots, {F_N}^g $ be (the naturally shifted versions of) an assignment that satisfies $ \psi_{t_g,f_g}$.
In particular this means $x_g \in P_{t_{g}} \cup M_{f_{g}}$, $x_{g+1}-1
\in F_{f_g}$, and $\flk{f_g}{x_g}{x_{g+1}}$.
Then define $P_h, M_h, F_h$ as follows.  
Let $P_h$ be the union of all ${P_h}^g$,
$M_h$ include all ${M_h}^g$,
$F_h$ include all ${F_h}^g$.
Also let $M_k$ contain $x_0$
and $F_k$ contain $x_\ell$.
Finally  let $P_{\delta(q_k,b)}$ contain $\ell +1$ if $a \lessdot b$ or $a \doteq b$.

Then  we show that  $\psi_{i,k}$ is satisfied
by checking every subformula in $\varphi_\delta$,
$\varphi_{exist}$, $\varphi_{unique}$.
By the inductive hypothesis, all axioms are satisfied within every
$s_g$. Thus, we only have to prove that they are satisfied in
positions $x_g$, for $0 \le g \le \ell$.
The proof of satisfaction of most axioms in $\psi_{i,k}$ is
clerical. Thus, we consider only a meaningful sample thereof.

\begin{enumerate}

\item $\varphi_{push\_fw}$ is satisfied for $x=x_{g-1}$ and $y = x_g$ since
$\suc {f_{g-1}}{x_{g-1}}{x_g} \lor \flk {f_{g-1}}{x_{g-1}}{x_g} $
holds and $\delta(q_{f_{g-1}},c_g) = q_{t_g}$, $x_g \in P_{t_g}$.

\item $\varphi_{flush\_fw}$ is satisfied
for $\tree_{f_\ell,f_0} (0,1,x_\ell,x_{\ell+1})$ since $0 \in M_k$, $x_\ell \in F_k$,  
$\delta(q_{f_{\ell}},
  q_{f_0}) = q_k$.

\item  $\varphi_{push\_bw2}$ is satisfied for $x_g\in P_{t_g}$ and $x_{g-1} \avv x_g$ (if $s_{g-1} \ne
\epsilon$), since $\flk {f_{g-1}} {x_{g-1}} {x_g}$ and $\delta(q_{f_{g-1}}, c_g) = q_{t_g}$. 

\item $\varphi_{push\_bwM}$, $\varphi_{push\_bwF}$,
  $\varphi_{push\_bw}$ are satisfied for $\tree_{f_\ell, f_0} (0,1,x_\ell,x_{\ell+1})$ by $\delta(q_{f_\ell},
  q_{f_0}) = q_k$.

\item $\varphi_{push\_unique}$, and $\varphi_{flush\_unique}$ are
  satisfied because $\mathcal A$ is deterministic.

\end{enumerate}

\noindent Hence $asb \models \psi_{i,k}$. \qed
\end{proof}

\begin{lemma}\label{lem:f2a}
For every chain $\mathcal C$, $asb \models \psi_{i, k}$ implies that there exists a support 
$q_a \ourpath{s}{q_k}$ for $\mathcal C$ in $\mathcal A$.
\end{lemma}

\begin{proof}
Again, we prove the lemma by induction on the structure of chains. 
\smallskip

\noindent {\bf Base step }
First consider the induction bases with $s_g = \epsilon$ for every $g = 0, 1, \ldots, \ell$, i.e., $\chain a s b$ is a simple chain with $s = c_1 c_2 \cdots c_\ell$.
Let $asb \models \psi_{i,k}$. Hence there is a suitable assignment for $\ee, P_h,
M_h, F_h$ such that
$0 \in P_i 
\ \land \
\flk k 0 {\ee+1}
\ \land \
\varphi_{\delta}
\ \land \
\varphi_{exist}
\ \land \
\varphi_{unique}$ 
holds true. 
Clearly $\ee$ is $|s|$.
For every $g$, let $t_g$ be the index such that $g \in P_{t_g}$. 
Notice that $t_g$ is unique by $\varphi_{push\_unique}$ and in particular $t_0 = i$.
Hence $t_g$ is the unique index such that $\suc {t_g} g {g+1}$.
Then, by $\varphi_{push\_bw1}$ with $y = g < \ell$, we have $\delta(q_{t_g},
c_{g+1}) = q_{t_{g+1}}$.
Moreover, since $\flk k 0 {\ell+1} \land \tree_{t_\ell,t_0} (0,1,\ell,\ell+1)$, 
by $\varphi_{flush\_bw}$ we get $\delta(q_{t_\ell}, q_{t_0}) = q_k $. Hence we
have built a support like~\eqref{eq:simplechainq}.

\smallskip

\noindent {\bf Induction step }
Now consider the general case with $s = s_0 c_1 s_1 \dots c_\ell s_\ell$
and again consider the assignment for $P_h,M_h, F_h$ that satisfies
$\psi_{i,k}$.
For every $g$, let $t_g$ be the index such that $x_g \in P_{t_g}$, and notice
that $t_g$ is unique by $\varphi_{push\_unique}$;
in particular $t_0 = i$.
For $g = 0,1, \ldots, \ell$, since $x_g \avv x_{g+1} \lor x_{g+1} = x_g+1$,
let $f_g$ be the index such that $\flk {f_g} {x_g} {x_{g+1}} \lor \suc {f_g} {x_g} {x_{g+1}}$.
Notice that such $f_g$ is unique by $\varphi_{unique}$ (see
Remark~\ref{rem:unique}), moreover $s_g = \epsilon$ implies $f_g = t_g$.
Hence if $s_g \neq \epsilon$, we have $c_g s_g c_{g+1} \models \psi_{t_g,f_g}$ and, by the inductive hypothesis, there exists a support
$q_{t_g} \ourpath{s_g}{q_{f_g}}$ in $\mathcal A$.

For every $g = 0 < \ell$, since $f_g$ is unique, by applying
$\varphi_{push\_bw1}$ with $y = x_{g+1} $ we get $\delta(q_{f_g}, c_{g+1}) =
q_{t_{g+1}}$.
Moreover, since $\tree_{i,t_\ell} (x_0, x_1, x_\ell, x_{\ell+1}) \land \flk k {x_0} {x_{\ell+1}} 
$,
by $\varphi_{flush\_bw}$ we get $\delta(q_{t_\ell}, q_i) = q_k $. Hence we have
built a support like~\eqref{eq:compchainq} and this concludes the proof.
\qed
\end{proof}

\begin{proposition}
\label{prop:logic:f2a}
Let $(\Sigma, M)$ be an operator precedence alphabet and $\mathcal A$
be a Floyd automaton over $(\Sigma,M)$. Then there exists an
MSO$_{\Sigma, M}$ sentence $\varphi$ such that $L (\mathcal A)=
L(\varphi)$.
\end{proposition}

\begin{proof} 
  Let $\varphi$ be the MSO$_{\Sigma, M}$ sentence defined
  in~\eqref{eq:logic}. We show that $L (\mathcal A)= L(\varphi)$ by
  applying the previous lemmata.
Consider an accepting computation of $s$ in $\mathcal A$. Then there
exists a support $q_0 \ourpath{s}{q_k}$ for the chain $\chain \# s
\#$, with $q_k$ a final state; hence by Lemma~\ref{lem:a2f}, $\# s\# \
\models \ \psi_{0,k}$.
Vice versa, let $ s \in L(\varphi)$, then $\# s \# \ \models \
\psi_{0,k}$ with $q_k$ a final state; hence Lemma~\ref{lem:f2a} implies
that there exists a path $q_0 \ourpath{s}{q_j}$ and this concludes the proof.
\qed
\end{proof}

\section{Conclusions and future work}\label{sec:concl}

This paper somewhat completes a research path that began more than
four decades ago and was resumed only recently with new -and old-
goals. FL enjoy most of the nice properties that made regular
languages highly appreciated and applied to achieve decidability and,
therefore, automatic analysis techniques. In this paper we added to
the above collection the ability to formalize and analyze FL by means
of suitable MSO logic formulae. 

 New research topics, however,
stimulate further investigation. Here we briefly mention only two
mutually related ones. On the one hand, FA devoted to analyze strings should be
extended in the usual way into suitable transducers. They could be
applied, e.g. to translate typical mark-up languages such as XML,
HTML, Latex, \ldots into their end-user view. Such languages, which
motivated also the definition of VPL, could be classified as
``explicit parenthesis languages'' (EPL), i.e. languages whose
syntactic structure is explicitly apparent in the input string. On the
other hand, we
plan to start from the remark that VPL are characterized by a well
precise shape of the OPM \cite{Crespi-ReghizziM12} to characterize
more general classes of such EPL: for instance the language of Example
1 is such a language that is not a VPL, however. Another notable
feature of FL, in fact, is that they are suitable as well to parse
languages with implicit syntax structure such as most programming
languages as to analyze and translate EPL.

\bibliographystyle{splncs}
\bibliography{VPDbib}

\end{document}